\documentclass{article}
\usepackage{spconf,amsmath,graphicx,hyperref}
\usepackage{spconf,amsmath,graphicx,bm,amssymb}
\usepackage{amsfonts,amsthm,mathtools}
\usepackage{algorithm}
\usepackage{algorithmic}
\usepackage[caption=false,font=footnotesize]{subfig}
\usepackage{textcomp}
\usepackage{xcolor}
\usepackage{placeins}
\usepackage{balance}
\usepackage{cite}

\newtheorem{theorem}{Theorem}

\newcommand{\C}{\mathbb C}

\newcommand{\INPUT}{\item[\textbf{Input:}]}
\newcommand{\OUTPUT}{\item[\textbf{Output:}]}

\title{A Low-Dimensional Riemannian Alternating Optimization Algorithm for BD-RIS-Assisted Multiuser Beamforming}

\name{Weitao Xia$^{\star}$, Zheyu Wu$^{\S}$, Ya-Feng Liu$^{\star}$
}
\address{$^{\star}$School of Mathematical Sciences, Beijing University of Posts and Telecommunications, Beijing, China \\
	$^{\S}$Department of Electrical and Electronic Engineering, Imperial College London, London, U.K.\\
   Email: wtxia@bupt.edu.cn, zheyu.wu@imperial.ac.uk, yafengliu@bupt.edu.cn
}

\begin{document}

\maketitle
\ninept

\begin{abstract}
	This paper investigates joint active and passive beamforming for a multiuser downlink system assisted by a fully-connected beyond-diagonal reconfigurable intelligent surface (BD-RIS). Since the number of tunable components grows quadratically with the number of RIS elements, directly optimizing the full scattering matrix becomes increasingly computationally expensive. To address this challenge, we first derive an equivalent low-dimensional formulation that replaces the full scattering matrix with a variable whose number of columns equals the number of users. We further establish a smooth manifold structure on a dense subset of the equivalent feasible set, based on which we develop a low-dimensional Riemannian alternating optimization (LD-RAO) algorithm, where the passive variable is updated using a single Riemannian gradient ascent step at each iteration. The convergence of LD-RAO is established under mild assumptions. Numerical results demonstrate that the proposed algorithm achieves the same performance as the existing state-of-the-art (SOTA) full-dimensional methods with significantly less CPU time, particularly for large-scale BD-RIS.
\end{abstract}

\begin{keywords}
Beyond-diagonal reconfigurable intelligent surface, multiuser beamforming, Riemannian optimization.
\end{keywords}

\section{Introduction}

Reconfigurable intelligent surface (RIS) offers an energy-efficient way to shape wireless propagation using nearly passive elements \cite{CONVENTIONAL_RIS, ConventionalRIS2, RISMIMO, LiuRIS}. In a conventional single-connected RIS, each element is connected to an independently tunable load \cite{BDRISNetwork}, yielding a diagonal scattering matrix that cannot redistribute signal energy among different RIS ports. Beyond-diagonal RIS (BD-RIS) generalizes this architecture with reconfigurable inter-element connections, thereby enabling non-diagonal scattering matrices \cite{BDRISNetwork, BDRISArchitecture}, and providing additional degrees of freedom to enhance desired signal power, mitigate interference, and improve coverage \cite{RISTutorial}.

Despite these advantages, the more complicated architecture of BD-RIS introduces new optimization challenges. For an $N$-element fully-connected BD-RIS, the number of tunable components grows quadratically with $N$, leading to a high problem dimension. Moreover, losslessness and reciprocity of BD-RIS require a unitary and symmetric scattering matrix \cite{BDRISNetwork}. These matrix constraints differ fundamentally from the elementwise unit-modulus constraints of conventional RIS and therefore require new optimization strategies tailored to BD-RIS. Accordingly, considerable research effort has been devoted to BD-RIS optimization \cite{RISTutorial, R3, BDRISClosedForm, FPPart1, BDRISUnified, R2, TwoStage, FPPSLA, ZheyuAlgorithm, R1, LossyBDRIS}.

In particular, joint active and passive beamforming for multiuser systems has been addressed using penalty dual decomposition \cite{BDRISUnified}, manifold optimization \cite{R2}, a heuristic Two-Stage method \cite{TwoStage}, projection-based algorithm \cite{FPPSLA}, and an ADMM-based approach \cite{ZheyuAlgorithm}. However, most of these approaches operate directly on the full $N\times N$ scattering matrix and typically involve iterative matrix factorizations or projections, whose computational costs increase rapidly with $N$. To reduce this complexity, the Two-Stage method in \cite{TwoStage} first designs the BD-RIS scattering matrix heuristically and then optimizes the active beamformer with the scattering matrix fixed. Although computationally efficient, this decoupled design does not directly optimize the original sum-rate objective and may therefore incur a performance loss. Hence, it remains challenging to develop an algorithm that achieves high performance while avoiding repeated optimization of the full $N\times N$ scattering matrix. Recently, the work \cite{LowdimProblem} provided a thorough characterization of the effective channel achievable by a fully-connected BD-RIS, showing that, although the effective channel is naturally expressed in terms of the full scattering matrix, its dependence on the BD-RIS can be fully characterized by a lower-dimensional matrix whose column dimension is determined by the number of users, which is typically much smaller than $N$. This characterization provides a promising way to address the aforementioned challenge by enabling BD-RIS optimization to be performed directly over this lower-dimensional representation rather than the full scattering matrix.

Motivated by the above observations, this paper reformulates the joint active and passive beamforming problem in terms of a low-dimensional BD-RIS variable, thereby significantly decreasing the problem dimension. 
However, the constraint structure of the new formulation remains difficult to handle. To address this issue, we introduce auxiliary variables to reformulate the problem with a more tractable constraint.
We then propose a Riemannian alternating optimization algorithm, in which the passive variable is updated via a tailored Riemannian gradient ascent step. The computational complexity of the proposed algorithm scales only linearly with $N$. Numerical results show that its sum-rate performance is nearly identical to that of state-of-the-art (SOTA) methods \cite{TwoStage, FPPSLA, ZheyuAlgorithm} while requiring substantially less computational time. \vspace{-2mm}

\section{Problem Formulation}\label{sec:problemformulation}
We consider a fully-connected BD-RIS-assisted downlink system comprising an $L$-antenna base station
(BS), an $N$-element BD-RIS, and $K$ single-antenna users. 
The direct channels from the BS to users are assumed to be blocked for simplicity. 
Let $\mathbf{W}=[\mathbf{w}_1,\mathbf{w}_2,\ldots,\mathbf{w}_K]\in\mathbb C^{L\times K}$, $\mathbf{G}\in\mathbb C^{N\times L}$, and $\mathbf{H}=[\mathbf{h}_1, \mathbf{h}_2,\ldots,\mathbf{h}_K]\in\mathbb C^{N\times K}$ denote the beamforming matrix, the channel from the BS to the BD-RIS, and the channels from the BD-RIS to the users, respectively.
For an ideal lossless and reciprocal BD-RIS, its scattering matrix
$\mathbf{\Theta}\in\mathbb C^{N\times N}$ satisfies
$\mathbf{\Theta}^H\mathbf{\Theta}=\mathbf{I}_N$ and $\mathbf{\Theta}^T=\mathbf{\Theta}$ \cite{BDRISNetwork}.
With the notations introduced above, the BD-RIS-assisted joint active and passive beamforming problem is formulated as 
\vspace{-1mm}
\begin{equation}
	\begin{aligned}
		\max_{\mathbf{W},\mathbf{\Theta}}~~
		&\sum_{k=1}^{K}\log\left(1+\frac{|\mathbf{h}_k^H\mathbf{\Theta} \mathbf{G}\mathbf{w}_k|^2}
		{\sum_{j\ne k}|\mathbf{h}_k^H\mathbf{\Theta} \mathbf{G}\mathbf{w}_j|^2+\sigma_k^2}\right)\\
		\mathrm{s.t.}~~
		&\mathbf{\Theta}^H\mathbf{\Theta}=\mathbf{I}_N,~~ \mathbf{\Theta}^T=\mathbf{\Theta},~~\|\mathbf{W}\|_F^2\le P_{\mathrm{T}} ,
	\end{aligned}
	\label{p:original_SR_problem}\vspace{-1mm}
\end{equation}
where $\sigma_k^2$  is the power of the additive white Gaussian noise (AWGN) at user $k$ and $P_{\mathrm{T}}$ is the BS transmit-power budget. In practice, $N$ is typically large and satisfies $N\gg K$, making it challenging to solve \eqref{p:original_SR_problem} efficiently. This paper aims to develop a low-complexity algorithm that achieves high sum-rate performance, particularly for large $N$. Throughout the paper, we focus on the large $N$ and assume $N\geq 2K$.

\vspace{-3mm}
\section{Proposed LD-RAO Algorithm}\label{sec:algorithm}
\vspace{-2mm}

In this section, we first establish a low-dimensional reformulation of \eqref{p:original_SR_problem} in Section~\ref{sec:lowdim_reformulation}.  We then present the overall algorithmic framework and the Riemannian gradient ascent update for the passive variable in Sections \ref{sec:algorithm_structure} and \ref{sec:RGA}, respectively. Finally, we  analyze the computational complexity and establish the convergence of the proposed algorithm in Section~\ref{sec:theoretical_analysis}. 
\vspace{-2mm}

\subsection{Low-Dimensional Reformulation}\label{sec:lowdim_reformulation}
The scattering matrix $\mathbf{\Theta}$ affects \eqref{p:original_SR_problem} only
through the effective channels $\mathbf{h}_k^H\mathbf{\Theta} \mathbf{G}$, $k = 1, 2, \ldots, K$. Hence, define \vspace{-1mm}
\begin{equation}
	\mathbf{u}_k=\mathbf{\Theta}^H\mathbf{h}_k,~~
	\mathbf{U}=[\mathbf{u}_1,\mathbf{u}_2,\ldots,\mathbf{u}_K]=\mathbf{\Theta}^H\mathbf{H},
	\label{eq:u_transformation}\vspace{-1mm}
\end{equation}
so that $\mathbf{h}_k^H\mathbf{\Theta} \mathbf{G}\mathbf{w}_j=\mathbf{u}_k^H\mathbf{G}\mathbf{w}_j$. Moreover, according to \cite{LowdimProblem}, the symmetry and unitarity of $\mathbf{\Theta}$ imply that $\mathbf{U}$ satisfies \vspace{-1mm}
\begin{equation}
	\mathbf{U}^H\mathbf{U}=\mathbf{H}^H\mathbf{H},~~ \mathbf{U}^T\mathbf{H}=(\mathbf{U}^T\mathbf{H})^T.
	\label{eq:U_constraints}\vspace{-1mm}
\end{equation}
By substituting \eqref{eq:u_transformation} into \eqref{p:original_SR_problem}, problem \eqref{p:original_SR_problem} can be written as\vspace{-2mm}
\begin{equation}
	\begin{aligned}
		\max_{\mathbf{W},\mathbf{U}}~~ &
		\sum_{k=1}^{K}\log\!\left(
		1+\frac{|\mathbf{u}_k^H\mathbf{G}\mathbf{w}_k|^2}
		{\sum_{j\ne k}|\mathbf{u}_k^H\mathbf{G}\mathbf{w}_j|^2+\sigma_k^2}\right)\\
		\mathrm{s.t.}~~&
		\eqref{eq:U_constraints}
		,~~\|\mathbf{W}\|_F^2\leq P_{\mathrm{T}}.
	\end{aligned}
	\label{eq:P_U}\vspace{-2mm}
\end{equation}
More precisely, let
$\mathcal U_{\mathbf{\Theta}}=\left\{\mathbf{\Theta}^H\mathbf{H}:\mathbf{\Theta}^H\mathbf{\Theta}=\mathbf{I}_N,\mathbf{\Theta}^T=\mathbf{\Theta}\right\}$
and $\mathcal{U}_{\mathrm{c}}$ denote the set characterized by \eqref{eq:U_constraints}. 
The two sets $\mathcal{U}_{\mathrm{c}}$ and $\mathcal{U}_{\mathbf{\Theta}}$ differ only by a measure-zero set \cite[Theorem~2]{LowdimProblem}.
Thus, \eqref{eq:P_U} is almost-everywhere equivalent to the original BD-RIS formulation \eqref{p:original_SR_problem}. 
Comparing \eqref{p:original_SR_problem} and \eqref{eq:P_U}, we see that the $N\times N$ variable $\mathbf{\Theta}$ is replaced with the $N\times K$ variable $\mathbf{U}$. Since $N\gg K$ in practice, \eqref{eq:P_U} substantially reduces the problem dimension. 

The constraints in \eqref{eq:U_constraints} remain challenging to handle directly. Next, we reformulate them into a more tractable form by representing the passive variable in an appropriate orthogonal basis.

Let $\mathbf{H}=\mathbf{P}\mathbf{\Sigma} \mathbf{Q}^H$ be the singular value decomposition (SVD) of $\mathbf{H}$, where $\mathbf{P}\in\C^{N\times K}$, $\mathbf{\Sigma} \in\C^{K\times K}$, and $\mathbf{Q}\in\C^{K\times K}$. Throughout the paper, we assume $\mathbf{H}$ is of full column rank without loss of generality.  
With the SVD of $\mathbf{H}$, the unitary constraint in \eqref{eq:U_constraints} can be written as $\mathbf{\Sigma}^{-1}\mathbf{Q}^H\mathbf{U}^H\mathbf{U}\mathbf{Q}\mathbf{\Sigma}^{-1} = \mathbf{I}_K$. Therefore, by defining $\mathbf{Y}= \mathbf{U}\mathbf{Q}\mathbf{\Sigma}^{-1}\in\C^{N\times K}$, or equivalently, $\mathbf{U}= \mathbf{Y} \mathbf{\Sigma} \mathbf{Q}^H$, the constraints in \eqref{eq:U_constraints} can be written as\vspace{-2mm}
\begin{equation}
\mathbf{Y}^H\mathbf{Y}=\mathbf{I}_K,~~ \mathbf{Y}^T\mathbf{P}=\mathbf{P}^T\mathbf{Y}.
\label{eq:Y_constraint}\vspace{-2mm}
\end{equation}
To express the symmetry constraint in \eqref{eq:Y_constraint} in a standard form, define
\(\mathbf{T}=\mathbf{P}^T\mathbf{Y}\in\mathbb{C}^{K\times K}\), yielding \(\mathbf{T}^T=\mathbf{T}\).
Let \((\cdot)^*\) denote the entrywise complex conjugate. Since \(\mathbf{P}^*\mathbf{T}
=\mathbf{P}^*\mathbf{P}^T\mathbf{Y}\), $\mathbf{T}$ represents the components of
\(\mathbf{Y}\) in \(\operatorname{span}(\mathbf{P}^*)\). 
For the remaining components of \(\mathbf{Y}\), we extend \(\mathbf{P}\) to a
unitary matrix
\(\mathbf{P}_{\mathrm{u}} \triangleq \left[\mathbf{P},\mathbf{P}_{{\mathrm{c}}}\right]\), where
\(\mathbf{P}_{{\mathrm{c}}}\in\mathbb{C}^{N\times m}\) is an orthonormal complement of $\mathbf{P}$ and \(m= N-K\).
Accordingly, every \(\mathbf{Y}\in\mathbb{C}^{N\times K}\) admits the unique
decomposition
\vspace{-2mm}
\begin{equation}
\mathbf{Y}
=\mathbf{P}^*\mathbf{T}
+\mathbf{P}_{\mathrm{c}}^*\mathbf{Z}
=
\mathbf{P}_{\mathrm{u}}^*\mathbf{X},
~~
\mathbf{X}
= \left[\mathbf{T}^T,\mathbf{Z}^T\right]^T,
\label{eq:Y_UcX}\vspace{-2mm}
\end{equation}
where \(\mathbf{Z}\in\mathbb{C}^{m\times K}\) represents the components of \(\mathbf{Y}\) in \(\operatorname{span}(\mathbf{P}_{\mathrm{c}}^*)\).
Hence, the passive variable is parameterized by \(\mathbf{X}\) with respect to the orthonormal basis \(\mathbf{P}_{\mathrm{u}}^*=[\mathbf{P}^*,\mathbf{P}_{\mathrm{c}}^*]\). 

Since $\mathbf{P}_{\mathrm{u}}^*$ is unitary,
the two constraints in \eqref{eq:Y_constraint} become $\mathbf{X}^H\mathbf{X}=\mathbf{I}_K$ and $\mathbf{T}^T=\mathbf{T}$, respectively. 
With the above variable transformation, \eqref{eq:P_U} is equivalently reformulated as\vspace{-1mm}
\begin{equation}
	\begin{aligned}
		\max_{\mathbf{W},\mathbf{X}}~~ &
		\sum_{k=1}^{K}\log\!\left(1+\frac{|\mathbf{u}_k^H\mathbf{G}\mathbf{w}_k|^2}
		{\sum_{j\ne k}|\mathbf{u}_k^H\mathbf{G}\mathbf{w}_j|^2+\sigma_k^2}\right)\\
		\mathrm{s.t.} ~~ &
		\mathbf{X}^H\mathbf{X}=\mathbf{I}_K,~~ \mathbf{T}^T=\mathbf{T},~~ \|\mathbf{W}\|_F^2\le P_{\mathrm{T}},\\
		&\mathbf{U}=\mathbf{P}_{\mathrm{u}}^*\mathbf{X}\mathbf{\Sigma}\mathbf{Q}^H.
	\end{aligned}
	\label{eq:P_X}\vspace{-2mm}
\end{equation}
Thus, the channel-dependent constraints \eqref{eq:U_constraints} are replaced by a semi-unitary constraint on $\mathbf{X}$ and a symmetry constraint on its upper block $\mathbf{T}$.
Combining the preceding equivalent transformations, \eqref{eq:P_X} constitutes an equivalent low-dimensional reformulation of \eqref{p:original_SR_problem}. 
We next focus on solving problem \eqref{eq:P_X}.

\vspace{-2mm}
\subsection{Algorithm Structure}\label{sec:algorithm_structure}
By using the Lagrangian-dual and quadratic
transforms of fractional programming (FP) \cite{FPPart1}, we can equivalently express the sum-rate objective function of \eqref{eq:P_X} as\vspace{-2mm}
\begin{equation}
	\begin{aligned}
		\max_{\bm{\alpha},\bm{\beta}}~~\sum_{k=1}^{K}\Bigg(& \log(1+\alpha_k)
		+2\sqrt{1+\alpha_k}\operatorname{Re}\!\left\{\overline{\beta_k}\mathbf{u}_k^{H}\mathbf{G} \mathbf{w}_k\right\}\\
		& -\alpha_k-|\beta_k|^2\left(\sum_{j=1}^{K}|\mathbf{u}_k^{H}\mathbf{G} \mathbf{w}_j|^2+\sigma_k^2\right) \Bigg),
	\end{aligned}
	\label{eq:fp_reform}\vspace{-2mm}
\end{equation}
where $\bm{\alpha} = \left[\alpha_1,\alpha_2,\ldots,\alpha_K\right]^T\in\mathbb{R}_+^K$ and $\bm{\beta} = \left[\beta_1,\beta_2,\ldots,\beta_K\right]^T\in\mathbb{C}^K$ are the auxiliary variables. 

We apply alternating optimization \cite{LiuSurvey} to maximize \eqref{eq:fp_reform} subject to the constraints in \eqref{eq:P_X}. At each iteration, the variables are updated in the order 
\[
(\bm{\alpha},\bm{\beta})\rightarrow \mathbf{W}\rightarrow (\bm{\alpha},\bm{\beta})\rightarrow \mathbf{X}.
\]
The auxiliary variables \((\bm{\alpha},\bm{\beta})\) are updated twice to ensure the surrogate function \eqref{eq:fp_reform} accurately reflects the objective function of \eqref{eq:P_X} in both objective value and first-order derivatives with respect to \(\mathbf W\) and \(\mathbf X\). In particular, we refresh the auxiliary variables after updating $\mathbf{W}$ to ensure the monotonic ascent property required by the convergence analysis. 
The updates of \(\bm{\alpha}\), \(\bm{\beta}\), and \(\mathbf{W}\) follow the procedures in \cite[Section IV-C]{FPPart1} and are therefore omitted for brevity. We next focus on the update of \(\mathbf{X}\).
\vspace{-3mm}
 
\subsection{Riemannian Gradient Ascent Update for the $\mathbf{X}$-Subproblem}
\label{sec:RGA}

With $(\bm{\alpha},\bm{\beta},\mathbf{W})$ fixed, the $\mathbf{X}$-subproblem takes the following form:\vspace{-1mm}
\begin{equation}
\hspace{-2mm}
	\begin{aligned}
		\max_{\mathbf{X}}~~ &
		f(\mathbf{X})\triangleq 2\operatorname{Re}\left\{\operatorname{tr}\!\left(\mathbf{X}^H\mathbf{D}\right)\right\}
		-\operatorname{tr}\!\left((\mathbf{A}^H\mathbf{X})\mathbf{S}(\mathbf{A}^H\mathbf{X})^H\right)\\
		\mathrm{s.t.}~~ & \mathbf{X}^H\mathbf{X}=\mathbf{I}_K,~~ \mathbf{T}^T=\mathbf{T},
	\end{aligned}
	\hspace{-2mm}
	\label{eq:X_sub_1}\vspace{-1mm}
\end{equation}
where $\mathbf{A}=\mathbf{P}_{\mathrm{u}}^T\mathbf{G}\mathbf{W}$,
$\mathbf{D}=\mathbf{A}\mathbf{\Sigma}_1\mathbf{Q}\mathbf{\Sigma}^H$,
$\mathbf{S}=\mathbf{\Sigma}\mathbf{Q}^H\mathbf{\Sigma}_2\mathbf{Q}\mathbf{\Sigma}^H$,
$\mathbf{\Sigma}_1=\operatorname{diag}\left(\overline{\beta_1}\sqrt{1+\alpha_1},\overline{\beta_2}\sqrt{1+\alpha_2},\ldots,\overline{\beta_K}\sqrt{1+\alpha_K}\right)$, and 
$\mathbf{\Sigma}_2=\operatorname{diag}\left(|\beta_1|^2,|\beta_2|^2,\ldots,|\beta_K|^2\right)$. 

In the rest of this subsection, we derive a Riemannian gradient ascent (RGA) update for \eqref{eq:X_sub_1}. To reduce the computational complexity, we perform only one RGA update per outer iteration, with the step size selected to satisfy the sufficient ascent condition. We next characterize the constraint manifold and derive the Riemannian gradient, step size, and retraction, which are important components of the proposed RGA. \vspace{-2mm}

\subsubsection{Constraint Manifold}
\label{part:manifold}

Let $\mathcal{F}$ denote the feasible set of problem \eqref{eq:X_sub_1}. To establish a smooth manifold structure suitable for Riemannian optimization, we consider the subset
\begin{equation*}
	\mathcal{F}^{\circ}\triangleq
	\left\{
	\mathbf{X}=\begin{bmatrix}\mathbf{T}\\\mathbf{Z}\end{bmatrix}:
	\mathbf{X}^H\mathbf{X}=\mathbf{I}_K,~~\mathbf{T}^T=\mathbf{T},~~\mathbf{I}_K-\mathbf{T}^H\mathbf{T}\succ \mathbf{0}
	\right\}.
\end{equation*}
It is apparent that $\overline{\mathcal{F}^{\circ}}=\mathcal{F}$, so we have\vspace{-2mm}
\begin{equation*}
	\max_{\mathbf{X}\in\mathcal{F}}f(\mathbf{X})
	=
	\sup_{\mathbf{X}\in\mathcal{F}^{\circ}}f(\mathbf{X}).
	\vspace{-2mm}
\end{equation*}
Therefore, without losing optimality, the proposed RGA is executed on $\mathcal{F}^{\circ}$. 
We next prove that $\mathcal{F}^{\circ}$ is a smooth manifold.
\begin{theorem}
	For $m\ge K$, $\mathcal F^{\circ}$ is a nonempty smooth embedded manifold
	of real dimension $2mK+K$, with the tangent space
	\begin{equation}
		T_{\mathbf{X}}\mathcal F^{\circ}=
		\left\{
		\bm{\eta}=\begin{bmatrix}\bm{\eta}_{\mathrm{T}} \\ \bm{\eta}_{\mathrm{Z}}\end{bmatrix}:
		\bm{\eta}_{\mathrm{T}}^T=\bm{\eta}_{\mathrm{T}},~~ \mathbf{X}^H\bm{\eta}+\bm{\eta}^H\mathbf{X}=\mathbf{0}
		\right\}.
		\label{eq:tangent_compact}
	\end{equation}
	\label{thm:manifold}
\end{theorem} \vspace{-8mm}
\begin{proof}
We prove this theorem using the regular level set theorem\cite[Corollary 5.14]{ManifoldTextbook}. 

Define
$\mathcal L^{\circ}=
\left\{
\mathbf{X}=\begin{bmatrix}\mathbf{T} \\ \mathbf{Z}\end{bmatrix}:
\mathbf{T}^T=\mathbf{T},~~ \mathbf{I}_K-\mathbf{T}^H\mathbf{T}\succ \mathbf{0}
\right\}$. It is simple to check that $\mathcal{L}^\circ$ is a smooth manifold of real dimension $K(K+1)+2mK$. 
Define
\[
F:\mathcal L^\circ\rightarrow\operatorname{Herm}(K),
~~
F(\mathbf{X})=\mathbf{X}^H\mathbf{X}-\mathbf{I}_K,
\] 
where $\operatorname{Herm}(K) = \left\{\mathbf{A}\in\mathbb{C}^{K\times K}: \mathbf{A}^H=\mathbf{A}\right\}$ is also a smooth manifold with real dimension $K^2$. 
Thus, $\mathcal F^{\circ}=F^{-1}(\mathbf{0})$. We next prove that $\mathbf{0}$ is a regular value of $F$. 
It suffices to prove that, for every $\mathbf{X}\in\mathcal{F}^{\circ}$, the differential 
\[
DF(\mathbf{X}): T_{\mathbf X}\mathcal{L}^{\circ}\rightarrow T_{F(\mathbf{X})}\operatorname{Herm}(K),\]
\[DF(\mathbf{X})[\bm{\eta}]=\mathbf{X}^H\bm{\eta}+\bm{\eta}^H\mathbf{X},
\] 
is surjective, where $T_{\mathbf X}\mathcal{L}^{\circ} = \left\{
\bm{\eta}=\begin{bmatrix}\bm{\eta}_{\mathrm{T}} \\ \bm{\eta}_{\mathrm{Z}}\end{bmatrix}:
\bm{\eta}_{\mathrm{T}}^T=\bm{\eta}_{\mathrm{T}}\right\}$ 
and $T_{F(\mathbf{X})}\operatorname{Herm}(K) = \operatorname{Herm}(K)$ \cite[Proposition 3.13]{ManifoldTextbook}.

At any fixed
$\mathbf{X}\in\mathcal F^{\circ}$, $\mathbf{Z}^H\mathbf{Z}=\mathbf{I}_K-\mathbf{T}^H\mathbf{T}\succ \mathbf{0}$. For any $\mathbf{H}_0\in\operatorname{Herm}(K)$, choose
$\bm{\eta}_{\mathrm{T}}=\mathbf{0}$,
$\bm{\eta}_{\mathrm{Z}}=\tfrac12 \mathbf{Z}(\mathbf{Z}^H\mathbf{Z})^{-1}\mathbf{H}_0$. Then $\bm{\eta} = \left[\bm{\eta}_{\mathrm{T}}^T, \bm{\eta}_{\mathrm{Z}}^T\right]^T$
gives $DF(\mathbf{X})[\bm{\eta}]=\mathbf{H}_0$. Therefore, $DF(\mathbf{X})$ is surjective, 
and the desired results follow.\vspace{-2mm}
\end{proof}

\subsubsection{Riemannian Gradient}
\label{part:riemannian_gradient}

Under the real Frobenius metric, the Euclidean gradient is $\mathbf{G}_{\mathrm{E}}=2\left(\mathbf{D}-\mathbf{A}(\mathbf{A}^H\mathbf{X})\mathbf{S}\right)$.
The Riemannian gradient $\mathbf{\Xi}$ is obtained by projecting the Euclidean gradient $\mathbf{G}_{\mathrm{E}}$ onto $T_{\mathbf{X}}\mathcal F^{\circ}$. Specifically, we first symmetrize the upper block of $\mathbf{G}_{\mathrm{E}}$ and then remove the corresponding normal component associated with the constraint $\mathbf{X}^H\boldsymbol{\eta}+\boldsymbol{\eta}^H\mathbf{X}=\mathbf{0}$ while maintaining the upper-block symmetry.

First, to guarantee the symmetry of the upper block, we define the projection operator $\Pi_{\mathrm{s}}(\cdot)$ as
\[
\Pi_{\mathrm{s}}\!\left(
\begin{bmatrix}
\mathbf{A}_{\mathrm{T}}\\
\mathbf{A}_{\mathrm{Z}}
\end{bmatrix}
\right)
=
\begin{bmatrix}
(\mathbf{A}_{\mathrm{T}}+\mathbf{A}_{\mathrm{T}}^T)/2\\
\mathbf{A}_{\mathrm{Z}}
\end{bmatrix},~~\mathbf{A}_{\mathrm{T}}\in\mathbb{C}^{K\times K},~~\mathbf{A}_{\mathrm{Z}}\in\mathbb{C}^{m\times K}.
\]
Applying this projection to the Euclidean gradient yields \(\mathbf{G}_{\mathrm{s}}=\Pi_{\mathrm{s}}(\mathbf{G}_{\mathrm{E}})\), which satisfies the upper-block symmetry constraint. Next, we impose the orthogonality condition $\mathbf{X}^H\boldsymbol{\eta}+\boldsymbol{\eta}^H\mathbf{X}=\mathbf{0}$, whose corresponding normal direction is represented by $\mathbf{X}\mathbf{\Lambda}$ with $\mathbf{\Lambda}=\mathbf{\Lambda}^H$. To maintain the upper-block symmetry while removing the normal component, we project $\mathbf{X}\mathbf{\Lambda}$ using $\Pi_{\mathrm{s}}(\cdot)$ before subtracting it from $\mathbf{G}_{\mathrm{s}}$. Hence the Riemannian gradient $\mathbf{\Xi}$ can be written as
\begin{equation}
	\mathbf{\Xi}
=
\mathbf{G}_{\mathrm{s}}-\Pi_{\mathrm{s}}(\mathbf{X}\mathbf{\Lambda}).
\label{eq:rgrad_compact}
\end{equation}
We next determine $\mathbf{\Lambda}$ by using the constraints in \eqref{eq:tangent_compact}.

The upper block of $\mathbf\Xi$ is symmetric by construction. Enforcing $\mathbf X^H\mathbf\Xi+\mathbf\Xi^H\mathbf X=\mathbf{0}$ yields the following linear equation \eqref{eq:multiplier_equation} for $\mathbf\Lambda$:\vspace{-1mm}
\begin{equation}
\Pi_{\mathrm{H}}\!\left(\mathbf{X}^H\Pi_{\mathrm{s}}(\mathbf{X}\mathbf{\Lambda})\right)
=\Pi_{\mathrm{H}}(\mathbf{X}^H\mathbf{G}_{\mathrm{s}})\triangleq \mathbf{B},
\label{eq:multiplier_equation}\vspace{-2mm}
\end{equation}
where $\Pi_{\mathrm{H}}(\mathbf{A}) = (\mathbf{A}^H+\mathbf{A})/2$. To solve \eqref{eq:multiplier_equation} in closed form, we apply the Takagi factorization
$\mathbf{T}=\mathbf{V}\mathbf{C}\mathbf{V}^T$, where $\mathbf{C}=\operatorname{diag}(c_1,c_2,\ldots,c_K)$, and express $\mathbf{\Lambda}$ as 
$\mathbf{\Lambda}=\mathbf{V}^*\mathbf{M}\mathbf{V}^T$, where $\mathbf{M}^H=\mathbf{M}$. Defining $\widetilde{\mathbf{B}}=\mathbf{V}^T\mathbf{B}\mathbf{V}^*$, we decouple \eqref{eq:multiplier_equation} entrywise, giving\vspace{-2mm}
\begin{equation}
\mathbf{M}_{ij}=
\frac{\operatorname{Re}(\widetilde{\mathbf{B}}_{ij})}{1-(c_i-c_j)^2/4}
+\mathrm{i}\frac{\operatorname{Im}(\widetilde{\mathbf{B}}_{ij})}{1-(c_i+c_j)^2/4},
\label{eq:projection_M}\vspace{-1mm}
\end{equation}
where $\mathrm{i}$ is the imaginary unit. Since \(0\leq c_i<1\), all denominators in \eqref{eq:projection_M} are positive, so \eqref{eq:projection_M} uniquely determines $\mathbf{M}$ and subsequently yields $\mathbf{\Lambda}$, which in turn determines $\mathbf{\Xi}$. Therefore, the resulting $\mathbf{\Xi}$ is precisely the orthogonal projection of $\mathbf{G}_{\mathrm{E}}$ onto the tangent space $T_{\mathbf{X}}\mathcal{F}^{\circ}$. Moreover, \(\langle\mathbf G_{\mathrm E},\mathbf\Xi\rangle=\|\mathbf\Xi\|_F^2\geq 0\), so \(\mathbf\Xi\) is an ascent direction whenever it is nonzero.

\vspace{-2mm}

\subsubsection{Step Size}
\label{part:step_retraction}

Let $\tau$ denote the step size. To guarantee the convergence, we impose the following sufficient ascent condition on $\tau$:\vspace{-1mm}
\begin{equation}
	f(\mathcal R_{\mathbf{X}}(\tau\mathbf{\Xi}))
	\ge f(\mathbf{X})+\delta \tau\|\mathbf{\Xi}\|_F^2,
	\label{eq:armijo_compact}\vspace{-1mm}
\end{equation}
where $\mathcal R_{\mathbf{X}}(\cdot) $ is the retraction defined in \eqref{eq:retraction_compact} and $0<\delta<1$. We apply the Armijo line search to find a step size satisfying \eqref{eq:armijo_compact}. Next, we discuss the initial step size $\widetilde{\tau}$. 

Substituting $\mathbf{X}+\tau\mathbf{\Xi}$ into $f(\mathbf{X})$ yields \vspace{-1mm}
\begin{equation}
		f(\mathbf{X}+\tau\mathbf{\Xi})=f(\mathbf{X})+\tau\|\mathbf{\Xi}\|_F^2-\tau^2\kappa,
	\label{eq:line_quadratic_compact}\vspace{-1mm}
\end{equation}
where $\kappa=\operatorname{tr}\!\left((\mathbf{A}^H\mathbf{\Xi})\mathbf{S}(\mathbf{A}^H\mathbf{\Xi})^H\right)\geq 0$. When $\kappa>0$, the quadratic function of $\tau$ in \eqref{eq:line_quadratic_compact} admits the unconstrained maximizer
\(\tau_{\rm quad}=\|\mathbf{\Xi}\|_F^2/(2\kappa)\). Meanwhile, to ensure that the retraction remains feasible, we introduce the retraction-domain bound $\tau_{\rm dom}=(1-\|\mathbf{T}\|_2)/\|\mathbf{\Xi}_{\mathrm{T}}\|_2$ with $\mathbf{\Xi} = \left[\mathbf{\Xi}_{\mathrm{T}}^T, \mathbf{\Xi}_{\mathrm{Z}}^T\right]^T$. Each of $\tau_{\rm quad}$ and $\tau_{\rm dom}$ is set to $+\infty$ when its respective denominator vanishes. We further introduce $0<\tau_{\max}<\infty$ to avoid excessively large and numerically unstable trial steps. Accordingly, the initial trial step size for Armijo backtracking is chosen as\vspace{-1mm}
\begin{equation}
		\widetilde{\tau}=\gamma\min\{\tau_{\rm quad},\tau_{\rm dom}, \tau_{\max}\},
	\label{eq:step_compact}\vspace{-1mm}
\end{equation} 
where $0 <\gamma <1$ is the damping parameter. Armijo backtracking selects the smallest nonnegative integer $j$ such that $\tau=\rho^j\widetilde{\tau}$ satisfies the sufficient ascent condition \eqref{eq:armijo_compact}, where $0<\rho<1$.

\vspace{-2mm}
\subsubsection{Retraction}\label{sec:retraction}

Once the step size $\tau$ is determined, we need to retract the tentative point $\widehat{\mathbf{X}} \triangleq \mathbf{X} + \tau\mathbf{\Xi}$ back to the manifold $\mathcal{F}^{\circ}$. Denote $\widehat{\mathbf{X}} = \left[ \widehat{\mathbf{T}}^T, \widehat{\mathbf{Z}}^T \right]^T$ and define $\mathbf{R}=(\mathbf{I}_K-\widehat{\mathbf{T}}^H\widehat{\mathbf{T}})^{1/2}.$
The retraction is\vspace{-3mm}
\begin{equation}
\mathcal R_{\mathbf{X}}(\tau\mathbf{\Xi})=
\begin{bmatrix}
\widehat{\mathbf{T}}\\
\widehat{\mathbf{Z}}\mathbf{R}\left(\mathbf{R}\widehat{\mathbf{Z}}^H\widehat{\mathbf{Z}}\mathbf{R}\right)^{-1/2}\mathbf{R}
\end{bmatrix}.
\label{eq:retraction_compact}\vspace{-1mm}
\end{equation}
Here we provide a brief proof that \eqref{eq:retraction_compact} defines a retraction. 
Since \(\tau<\tau_{\rm dom}\), \(\mathbf R\) is positive definite. In addition, because \(\mathbf\Xi\) is a tangent vector at \(\mathbf X\), it is simple to check that \(\widehat{\mathbf Z}\) has full column rank. Therefore, \(\mathcal R_{\mathbf X}\) is a well-defined, smooth map into \(\mathcal F^\circ\). Moreover, one can verify that \(\mathcal R_{\mathbf X}(\tau\mathbf\Xi)=\mathbf X+\tau\mathbf\Xi+O(\tau^2)\) and hence $\mathcal{R}_{\mathbf{X}}$ satisfies the two defining properties of a retraction. Details are omitted for brevity.

To summarize, at each RGA update, we compute the Riemannian gradient $\mathbf{\Xi}$ using \eqref{eq:rgrad_compact}, apply Armijo line search with the initial trial step \eqref{eq:step_compact} to select a step size $\tau$ satisfying \eqref{eq:armijo_compact}, and retract the tentative point $\widehat{\mathbf{X}} = \mathbf{X}+\tau\mathbf{\Xi}$ onto $\mathcal F^{\circ}$ through \eqref{eq:retraction_compact}. The overall LD-RAO algorithm is summarized in Algorithm~\ref{alg:workflow}.

\begin{algorithm}[t]
\caption{LD-RAO with one RGA update}
\label{alg:workflow}
\begin{algorithmic}[1]
\INPUT $\mathbf{W}^{(0)}$, $\mathbf{X}^{(0)}\in\mathcal{F}^{\circ}$, $\gamma,\rho,\delta,\tau_{\max}$.
\FOR{$t=0,1,\ldots$}
  \STATE Update $\bm{\alpha}^{(t,0)},\bm{\beta}^{(t,0)}$, $\mathbf{W}^{(t+1)}$.
  \STATE Refresh $\bm{\alpha}^{(t,1)},\bm{\beta}^{(t,1)}$ with $\mathbf{W}^{(t+1)}$.
  \STATE Form $\mathbf{A}^{(t)},\mathbf{D}^{(t)},\mathbf{S}^{(t)}$ using the definition following \eqref{eq:X_sub_1}.
  \STATE Compute $\mathbf{\Xi}^{(t)}$ by \eqref{eq:rgrad_compact}.
  \IF{$\|\mathbf{\Xi}^{(t)}\|_F=0$}
    \STATE $\mathbf{X}^{(t+1)}\leftarrow \mathbf{X}^{(t)}$.
  \ELSE
    \STATE Compute $\widetilde{\tau}^{(t)}$ from \eqref{eq:step_compact}; set $j\leftarrow0$.
    \REPEAT
      \STATE $\tau\leftarrow\rho^j\widetilde{\tau}^{(t)}$,
             $\mathbf{X}_{\rm try}\leftarrow\mathcal R_{\mathbf{X}^{(t)}}(\tau\mathbf{\Xi}^{(t)})$ by
             \eqref{eq:retraction_compact}.
             \STATE $j\leftarrow j+1$.
    \UNTIL{the sufficient ascent condition \eqref{eq:armijo_compact} holds}
    \STATE $\mathbf{X}^{(t+1)}\leftarrow \mathbf{X}_{\rm try}$.
  \ENDIF
  \IF{the stopping test holds}
    \STATE \textbf{break}
  \ENDIF
\ENDFOR
\STATE Recover $\mathbf{\Theta} = \mathbf{P}\left(\mathbf{X}^{(t+1)}\right)^H\mathbf{P}_{\mathrm{u}}^T$.
\OUTPUT $\mathbf{W}^{(t+1)},\mathbf{\Theta}$.
\end{algorithmic}
\end{algorithm}

\begin{table}[t]
	\centering\vspace{-5mm}
	\caption{Computational Complexity Comparison.}
	\label{tab:complexity}
	\renewcommand{\arraystretch}{1.2}
	\resizebox{\columnwidth}{!}{%
		\begin{tabular}{ccc}
			\hline
			\textbf{Algorithm} & \textbf{Reference} &
			\textbf{Computational Complexity} \\
			\hline\hline
			LD-RAO
			& Proposed
			& $\mathcal{O}\!\left(
			2I_{\mathrm{out}}\left(NK^2+K^3\right)\right)$ \\
			FP-PSLA
			& \cite{FPPSLA}
			& $\mathcal{O}\!\left(
			I_{\mathrm{out}}\left(I_{\mathbf{\Theta}}N^3+I_{\mathbf{W}}L^2K\right)\right)$ \\
			pp-ADMM
			& \cite{ZheyuAlgorithm}
			& $\mathcal{O}\!\left(I_{\rm out}\left(N^3K^3\right)
			\right)$ \\
			\hline
		\end{tabular}
	}\vspace{-5mm}
\end{table}

\vspace{-4mm}

\subsection{Theoretical Analysis}\label{sec:theoretical_analysis}\vspace{-1mm}
In this subsection, we present the convergence properties of the proposed algorithm and analyze its computational complexity. Specifically, we can show, under mild assumptions, that any limit point of the sequence generated by Algorithm~\ref{alg:workflow} is a stationary point of problem \eqref{eq:fp_reform}, which is equivalent to problem \eqref{eq:P_X}. The convergence result is primarily based on the sufficient ascent achieved after the update of each variable block. The detailed proof is omitted due to space limitations. 

Regarding computational complexity, Table~\ref{tab:complexity} summarizes the computational complexities of the proposed algorithm, FP-PSLA \cite{FPPSLA}, and pp-ADMM \cite{ZheyuAlgorithm}, while the Two-Stage method is omitted because of its relatively low sum-rate performance. Here, $I_{\mathrm{out}}$, $I_{\mathbf{\Theta}}$, and $I_{\mathbf{W}}$ denote the number of outer iterations, the number of passive inner-loop iterations, and the number of active inner-loop iterations, respectively. It can be observed that the computational complexity of LD-RAO scales linearly with $N$, whereas the dominant computational costs of FP-PSLA and pp-ADMM scale cubically with $N$, demonstrating the low-complexity nature of the proposed algorithm.

\vspace{-2mm}
\section{Numerical Results}\label{sec:numerical}\vspace{-2mm}

In this section, we present simulation results to demonstrate the effectiveness of the proposed algorithm, particularly for large $N$. 

\begin{figure}[t]
	\centering
	\includegraphics[width=0.42\textwidth]{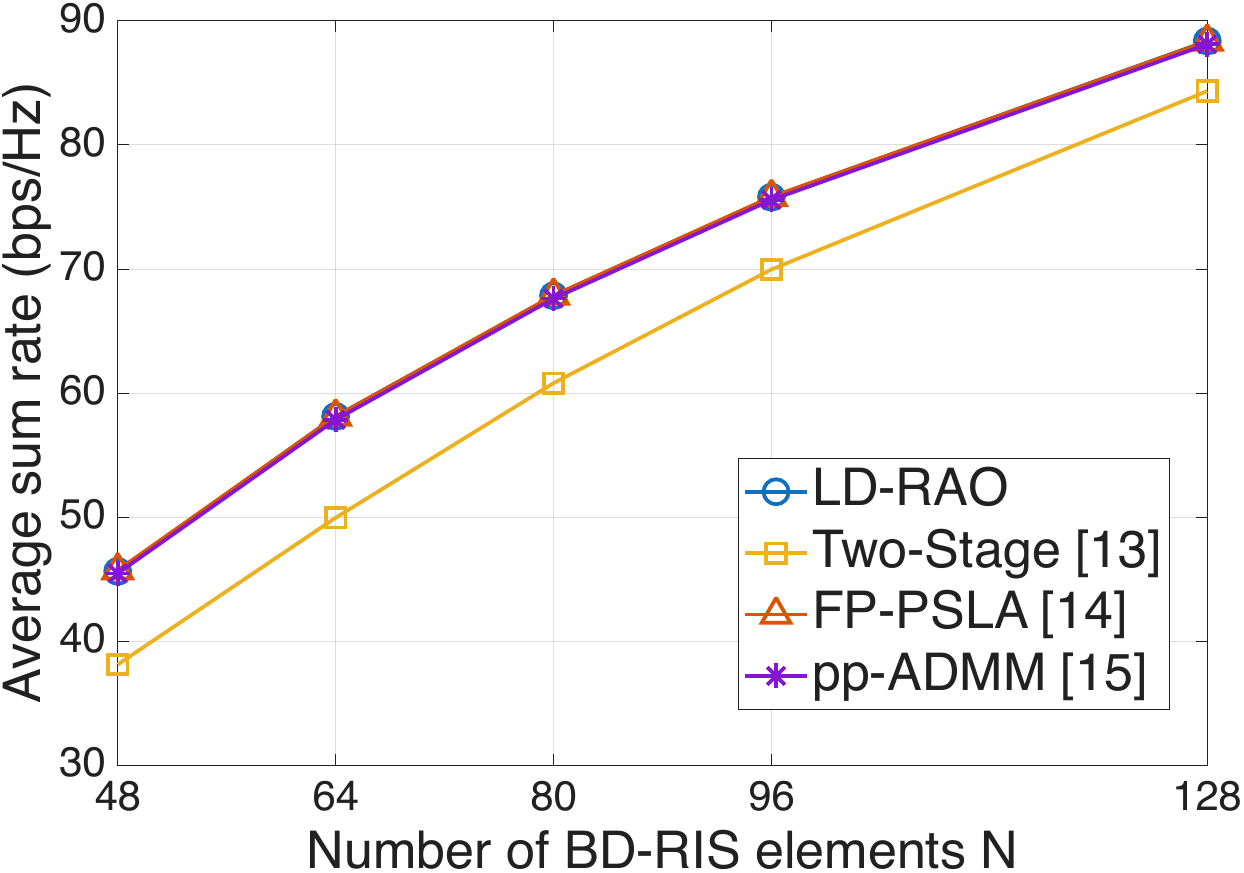}\vspace{-2mm}
	\caption{Average sum rate.}
	\label{fig:objective}\vspace{-2mm}
\end{figure}

\begin{figure}[t]
	\centering
	\includegraphics[width=0.42\textwidth]{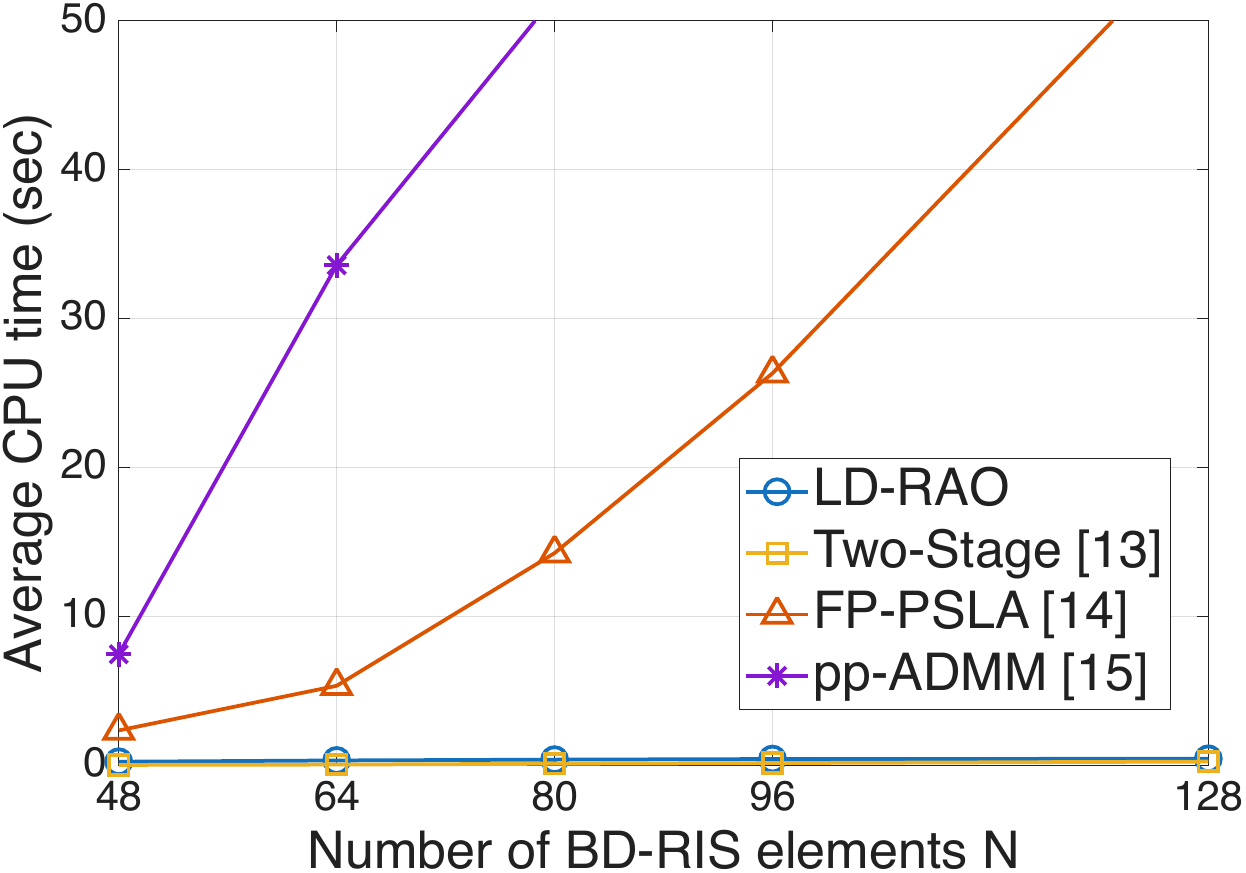}\vspace{-2mm}
	\caption{Average CPU time.}
	\label{fig:time}\vspace{-5mm}
\end{figure}

The simulation setup and initial point construction follow \cite{FPPSLA}, and we compare LD-RAO with Two-Stage~\cite{TwoStage}, FP-PSLA~\cite{FPPSLA}, 
and pp-ADMM~\cite{ZheyuAlgorithm} under the same 200 independent channel realizations with $L=32$, $K=16$, and $P_{\mathrm{T}}=50$~W. 

Figs.~\ref{fig:objective} and~\ref{fig:time} report the average sum rate and CPU time of the considered algorithms, respectively. As shown in Fig.~\ref{fig:objective}, LD-RAO achieves nearly the same sum rate as FP-PSLA and pp-ADMM, while all three consistently outperform the heuristic Two-Stage method. In terms of the CPU time, however, the algorithms exhibit significantly different scaling behaviors as $N$ increases. The CPU time of the proposed LD-RAO grows only slightly with $N$ and remains below $1\,\mathrm{s}$ even at $N=128$, whereas the CPU time of FP-PSLA and pp-ADMM rise sharply. 

To further explain this advantage, as shown in Table~\ref{tab:complexity}, the complexity of LD-RAO scales linearly with $N$, whereas the dominant computational costs of FP-PSLA and pp-ADMM scale cubically with $N$, making these methods increasingly expensive for large-scale BD-RIS, which is consistent with the CPU-time results observed in Fig.~\ref{fig:time}.

To conclude, the proposed LD-RAO approach achieves nearly identical performance to that of the SOTA methods with significantly lower computational complexity. This is attributed to both the low-dimensional formulation and the efficient algorithmic updates.

\bibliographystyle{IEEEbib}
\bibliography{reference.bib}

\end{document}